%% file: main.tex
\documentclass[english]{article}
\usepackage[latin9]{inputenc}
\usepackage{geometry}
\usepackage{array}
\usepackage{verbatim}
\usepackage{float}
\usepackage{multirow}
\usepackage{amsmath}
\usepackage{amsthm}
\usepackage{amssymb}

\makeatletter

\providecommand{\tabularnewline}{\\}
\floatstyle{ruled}
\newfloat{algorithm}{tbp}{loa}
\providecommand{\algorithmname}{Algorithm}
\floatname{algorithm}{\protect\algorithmname}

\numberwithin{equation}{section}
\numberwithin{figure}{section}
\theoremstyle{plain}
\newtheorem{thm}{\protect\theoremname}
  \theoremstyle{definition}
  \newtheorem{defn}[thm]{\protect\definitionname}
  \theoremstyle{plain}
  \newtheorem{lem}[thm]{\protect\lemmaname}

\pdfoutput=1

\usepackage{amsmath, amsthm, amssymb, amsfonts}
\usepackage{bbm}
\usepackage{bbold}
\usepackage[basic]{complexity}
\usepackage[pdftex,bookmarks,colorlinks]{hyperref}
\usepackage{enumitem}
\usepackage{color}

\makeatother

\usepackage{babel}
  \providecommand{\definitionname}{Definition}
  \providecommand{\lemmaname}{Lemma}
\providecommand{\theoremname}{Theorem}

\begin{document}
\include{general_notation}

\title{Subquadratic Submodular Function Minimization}

\author{Deeparnab Chakrabarty\thanks{Microsoft Research, India. Email: deeparnab@gmail.com}
\and Yin Tat Lee\thanks{Microsoft Research, Email: yile@microsoft.com}
\and Aaron Sidford\thanks{Stanford University. Email: sidford@stanford.edu}
\and Sam Chiu-wai Wong \thanks{UC Berkeley. Email: samcwong@berkeley.edu}}
\maketitle
\begin{abstract}
Submodular function minimization (SFM) is a fundamental discrete optimization
problem which generalizes many well known problems, has applications
in various fields, and can be solved in polynomial time. Owing to
applications in computer vision and machine learning, fast SFM algorithms
are highly desirable. The current fastest algorithms \cite{LSW15}
run in $O(n^{2}\log nM\cdot\time+n^{3}\log^{O(1)}nM)$ time and $O(n^{3}\log^{2}n\cdot\time+n^{4}\log^{O(1)}n$)
time respectively, where $M$ is the largest absolute value of the
function (assuming the range is integers) and $\time$ is the time
taken to evaluate the function on any set. Although the best known
lower bound on the query complexity is only $\Omega(n)$ \cite{Har08},
the current shortest non-deterministic proof \cite{Cun85} certifying
the optimum value of a function requires $\Omega(n^{2})$ function
evaluations. 

\medskip{}

The main contribution of this paper are \emph{subquadratic} SFM algorithms.
For integer-valued submodular functions, we give an SFM algorithm
which runs in $O(nM^{3}\log n\cdot\time)$ time giving the first \emph{nearly
linear} time algorithm in any known regime. For real-valued submodular
functions with range in $[-1,1]$, we give an algorithm which in $\tilde{O}(n^{5/3}\cdot\time/\varepsilon^{2})$
time returns an $\varepsilon$-additive approximate solution. At the
heart of it, our algorithms are projected stochastic subgradient descent
methods on the Lovasz extension of submodular functions where we crucially
exploit submodularity and data structures to obtain fast, i.e. sublinear
time subgradient updates. The latter is crucial for beating the $n^{2}$
bound \textendash{} we show that algorithms which access only subgradients
of the Lovasz extension, and these include the empirically fast Fujishige-Wolfe
heuristic \cite{Wolfe76,Fujishige80} and the theoretically best cutting
plane methods \cite{LSW15} , must make $\Omega(n)$ subgradient calls
(even for functions whose range is $\{-1,0,1\}$). 
\end{abstract}
\thispagestyle{empty}\newpage
\clearpage \setcounter{page}{1}

\input{intro.tex}

\input{preliminaries.tex}

\input{pseduopolynomial.tex}

\input{lower_bound.tex}

\section*{Acknowledgments}

This work was partially supported by NSF awards 0843915, 1111109,
CCF0964033, CCF1408635 and Templeton Foundation grant 3966. Part of
this work was done while the first three authors were visiting the
Hausdorff Research Institute for Mathematics in Bonn for the Workshop
on Submodularity, and the last three authors were visiting the Simons
Institute for the Theory of Computing in Berkeley. We thank the organizers
of the workshop for inviting us. We thank Elad Hazan and Dan Garber
for helpful preliminary discussions regarding approximate SFM. We
thank the anonymous reviewers for their helpful comments and in particular
for pointing us to needed references and previous work as well as
pointing us to the relationship between our work and graph optimization,
encouraging us to write Appendix~\ref{sec:minimum_cut}. A special
thanks to Bobby Kleinberg for asking the question about approximate
SFM.

\bibliographystyle{plain}
\bibliography{submodular}

\appendix
\include{appendix}

\end{document}

%% file: general_notation.tex
\global\long\def\R{\mathbb{R}}
 \global\long\def\Z{\mathbb{Z}}

\global\long\def\ellOne{\ell_{1}}
 \global\long\def\ellTwo{\ell_{2}}
 \global\long\def\ellInf{\ell_{\infty}}

\global\long\def\boldVar#1{\mathbf{#1}}
\global\long\def\mvar#1{\boldVar{#1}}
\global\long\def\vvar#1{\vec{#1}}


\global\long\def\defeq{\stackrel{\mathrm{{\scriptscriptstyle def}}}{=}}
\global\long\def\E{\mathbb{\mathbf{{E}}}}
\global\long\def\otilde{\tilde{O}}


\global\long\def\gradient{\bigtriangledown}
 \global\long\def\grad{\gradient}
 \global\long\def\hessian{\gradient^{2}}
 \global\long\def\hess{\hessian}
 \global\long\def\jacobian{\mvar J}
 \global\long\def\gradIvec#1{\vvar{f_{#1}}}
 \global\long\def\gradIval#1{f_{#1}}

 \global\long\def\setVec#1{\onesVec_{#1}}
 \global\long\def\indicVec#1{\onesVec_{#1}}

\global\long\def\specGeq{\succeq}
 \global\long\def\specLeq{\preceq}
 \global\long\def\specGt{\succ}
 \global\long\def\specLt{\prec}

\global\long\def\innerProduct#1#2{\big\langle#1 , #2 \big\rangle}
 \global\long\def\norm#1{\big\|#1\big\|}
 \global\long\def\normA#1{\norm{#1}_{\ma}}
 \global\long\def\normInf#1{\norm{#1}_{\infty}}
 \global\long\def\normOne#1{\norm{#1}_{1}}
 \global\long\def\normTwo#1{\norm{#1}_{2}}

\global\long\def\OPT{\mathrm{opt}}
\global\long\def\opt#1{#1^{\OPT}}

\global\long\def\fopt{f^{*}}

\global\long\def\ma{\mvar A}
 \global\long\def\mb{\mvar B}
 \global\long\def\mc{\mvar C}
 \global\long\def\md{\mvar D}
 \global\long\def\mf{\mvar F}
 \global\long\def\mg{\mvar G}
 \global\long\def\mh{\mvar H}
\global\long\def\mI{\mvar I}
 \global\long\def\mm{\mvar M}
 \global\long\def\mq{\mvar Q}
 \global\long\def\mr{\mvar R}
 \global\long\def\ms{\mvar S}
 \global\long\def\mt{\mvar T}
 \global\long\def\mU{\mvar U}
 \global\long\def\mv{\mvar V}
 \global\long\def\mw{\mvar W}
 \global\long\def\mx{\mvar X}
 \global\long\def\my{\mvar Y}
\global\long\def\mz{\mvar Z}
 \global\long\def\mproj{\mvar P}
 \global\long\def\mSigma{\mvar{\Sigma}}
 \global\long\def\mLambda{\mvar{\Lambda}}
  \global\long\def\mzero{\mvar 0}
\global\long\def\mlap{\mvar{\mathcal{L}}}
\global\long\def\mpi{\mvar{\mathcal{\Pi}}}

\global\long\def\eps{\mvar{\varepsilon}}

\global\long\def\mdiag{\mvar{\texttt{d}iag}}

\global\long\def\sign{\mvar{\mathsf{\mathsf{sign}}}}

\global\long\def\oracle{\mathcal{O}}
 \global\long\def\simplex{\Delta}

\global\long\def\abs#1{\left|#1\right|}

\global\long\def\ceil#1{\left\lceil #1 \right\rceil }

\global\long\def\time{\mathcal{\mathrm{{EO}}}}

\global\long\def\runtime{\mathcal{\time}}
 \global\long\def\timeOf#1{\runtime\left(#1\right)}

\global\long\def\domain{\mathcal{D}}

\global\long\def\argmin{\mathrm{argmin}}
\global\long\def\nnz{\mathrm{nnz}}
\global\long\def\vol{\mathrm{vol}}
\global\long\def\supp{\mathrm{supp}}
\global\long\def\dist{\mathcal{D}}

%% file: intro.tex
\section{Introduction}

Submodular functions are set functions that prescribe a value to every
subset of a finite universe $U$ and have the following diminishing
returns property: for every pair $S\subseteq T\subseteq U$, and for
every element $i\notin T$, $f(S\cup i)-f(S)\geq f(T\cup i)-f(T)$.
Such functions arise in many applications. For instance, the utility
functions of agents in economics are often assumed to be submodular,
the cut functions in directed graphs or hypergraphs are submodular,
the entropy of a given subset of random variables is submodular, etc.
 Submodular functions have been extensively studied for more than
five decades \cite{Choquet55,E70,Lovasz83,fujishigeB,Mccormick06}.

One of the most important problems in this area is submodular function
minimization (SFM, henceforth) which asks to find the set $S$ minimizing
$f(S)$. Note that submodular functions need not be monotone and therefore
SFM is non-trivial. In particular, SFM generalizes the minimum cut
problem in directed graphs and hypergraphs, and is a fundamental problem
in combinatorial optimization. More recently, SFM has found many
applications in areas such as image segmentation \cite{BVZ01,KKT08,KT10}
and speech analysis \cite{LB10,LB11}. Owing to these large scale
problems, fast SFM algorithms are highly desirable.

We assume access to an \emph{evaluation oracle }for the submodular
function, and use $\runtime$ to denote the time taken per evaluation.
An amazing property of submodular functions is that SFM can  be exactly
solved with polynomial many queries and in polynomial time. This was
first established via the ellipsoid algorithm \cite{GLS81} in 1981,
and the first polynomial combinatorial algorithms were obtained \cite{Cun85,IFF01,S00,IO09}
much later. 

The current fastest algorithms for SFM are by the second, third, and
fourth authors of this paper \cite{LSW15} who give $O(n^{2}\log nM\cdot\time+n^{3}\log^{O(1)}nM)$
time and $O(n^{3}\log^{2}n\cdot\time+n^{4}\log^{O(1)}n)$ time algorithms
for SFM. Here $M$ is the largest absolute value of the integer-valued
function. The former running time is a (weakly) polynomial running
time, i.e. it depends polylogarithmically on $M,$ while the latter
is a strongly polynomial running time, i.e. it does not depend on
$M$ at all. Although good in theory, known implementations of the
above algorithms are slow in practice \cite{FHI06,FI11,Bac13,CJK14}.
A different algorithm, the so-called Fujishige-Wolfe algorithm \cite{Wolfe76,Fujishige80}
seems to have the best empirical performance \cite{Bac13,JLB11,Bilmes15}
among general purpose SFM algorithms. Recently  the Fujishige-Wolfe
algorithm and variants were shown \cite{CJK14,LJ15} to run in $O((n^{2}\cdot\time+n^{3})M^{2})$
time, proving them to be \emph{pseudopolynomial} time algorithms,
that is having running time $O(\poly(n,\time,M))$.

In this paper we also consider approximate SFM. More precisely, for
submodular functions whose values are in the range $[-1,+1]$ (which
is without loss of generality by scaling), we want to obtain \emph{additive
}approximations\footnote{We also show in Appendix~\ref{app:mult_to_add} how to obtain a multiplicative
approximation under a mild condition on $f$. Such a condition is
necessary as multiplicative approximation is ill defined in general.}, that is, return a set $S$ with $f(S)\leq\textrm{\ensuremath{\OPT}}+\eps.$
Although approximate SFM has not been explicitly studied before, previous
works \cite{LSW15,Bac13,CJK14} imply $O(n^{2}\time\log^{O(1)}(n/\eps))$-time
and $O((n^{2}\cdot\time+n^{3})/\eps^{2})$-time algorithms. Table
\ref{tab:time} summarizes the above discussion. 

\begin{table}[h]
\centering{}%
\begin{tabular}{|c||c|c|c|}
\hline 
Regime & {\small{}Previous Best Running Time} & {\small{}Our Result} & {\small{}Techniques}\tabularnewline
\hline 
\hline 
{\small{}Strongly Polynomial} & \multicolumn{2}{c|}{{\small{}$O(n^{3}\log^{2}n\cdot\time+n^{4}\log^{O(1)}n)$ \cite{LSW15}}} & {\small{}Cutting Plane + Dimension Collapsing}\tabularnewline
\hline 
{\small{}Weakly Polynomial} & \multicolumn{2}{c|}{{\small{}$O(n^{2}\log nM\cdot\time+n^{3}\log^{O(1)}nM)$\cite{LSW15}}} & {\small{}Cutting Plane}\tabularnewline
\hline 
{\small{}Pseudo Polynomial} & {\small{}$O((n^{2}\cdot\time+n^{3})M^{2})$\cite{CJK14,LJ15}} & \textbf{\small{}$\tilde{{O}}(nM^{3}\cdot\time)$} & {\small{}See Section~}\ref{sec:technique_overview}\tabularnewline
\hline 
\multirow{1}{*}{{\small{}$\eps$-Approximate}} & \multirow{1}{*}{{\small{}$O(n^{2}\cdot\time/\varepsilon^{2})$ \cite{CJK14,LJ15,Bac13}}} & \textbf{\small{}$\tilde{O}(n^{5/3}\cdot\time/\varepsilon^{2})$} & {\small{}See Section~}\ref{sec:technique_overview}\tabularnewline
\hline 
\end{tabular}\caption{\label{tab:time} {\small{}Running times for minimizing a submodular
function defined on a universe of size $n$ that takes integer values
between $-M$ and $M$ (except for $\varepsilon$-approximate algorithms
we assume the submodular function is real-valued with range in $[-1,1]$).
$\text{EO}$ denotes the time to evaluate the submodular function
on a set.}}
\end{table}

In particular, the best known dependence on $n$ is \emph{quadratic}
even when the exact algorithms are allowed to be pseudopolynomial,
or when the $\eps$-approximation algorithms are allowed to have a
polynomial dependence on $\eps$. This quadratic dependence seems
to be a barrier. For exact SFM, the smallest known non-deterministic
proof \cite{E70,Cun85} that certifies optimality requires $\Theta(n^{2})$
queries, and even for the approximate case, nothing better is known
(see Appendix~\ref{app:approx_sfm_certificates}). Furthermore, in
this paper we \emph{prove} that a large class of algorithms which
includes the Fujishige-Wolfe algorithm\cite{Wolfe76,Fujishige80}
and the cutting planealgorithms of Lee et al.\cite{LSW15}, as stated
need to make $\Omega(n^{2})$ queries. More precisely, these algorithms
do not exploit the full power of submodularity and work even with
the weaker model of having access only to the ``subgradients of the
Lovasz Extension'' where each subgradient takes $\Theta(n)$ queries.
We prove that any algorithm must make $\Omega(n)$ subgradient calls
implying the quadratic lower bound for this class of algorithms. Furthermore,
our lower bound holds even for functions with range $\{-1,0,1\}$,
and so trivially the lower bound also holds for approximate SFM as
well.

\subsection{Our Results}

In this paper, we describe exact and approximate algorithms for SFM
which run in time subquadratic in the dimension $n$. Our first result
is a pseudopolynomial time exact SFM algorithm with\emph{ nearly linear
dependence }on $n$.  More precisely, for any integer valued submodular
function with maximum absolute value $M$, our algorithm returns the
optimum solution in $O(nM^{3}\log n\cdot\time)$ time.\textbf{ }This
has a few consequences to the complexity theory of SFM. First, this
gives a better dependence on $n$ for pseudopolynomial time algorithm.
Second, this shows that to get a super-linear lower bound on the query
complexity of SFM, one need to consider a function with super constant
function values.\footnote{Conversely, \cite[Thm 5.7]{harvey2008matchings} shows that we need
at least $n$ queries of evaluation oracle to minimize a submodular
function with range in $\{0,1,2\}$.} Third, this completes the following picture on the complexity of
SFM: the best known strongly polynomial time algorithms have query
complexity $\tilde{O}(n^{3})$, the best known (weakly) polynomial
time algorithms have query complexity $\tilde{O}(n^{2})$, and our
result implies the best pseudopolynomial time algorithm has query
complexity $\tilde{O}(n)$ .

Our second result is a \emph{subquadratic approximate SFM }algorithm\emph{.
}More precisely, we give an algorithm which in $\tilde{O}(n^{5/3}\time/\eps^{2})$
time, returns an $\eps$-additive approximate solution.  To break
the quadratic barrier, that arise from the need to compute $\Omega(n)$
subgradient each of which individully we do not know how to compute
faster than $\Omega(n\cdot\time)$, we wed continuous optimization
techniques  with properties deduced from submodularity and simple
data structures. These allow us to compute and use gradient updates
in a much more economical fashion. We believe that that the ability
to obtain subquadratic approximate algorithms for approximate submodular
minimization is an interesting structural result that could have further
implications.\footnote{Note that simple graph optimization problems, such as directed minimum
$s$-$t$ cut, is not one of these (See Appendix~\ref{sec:minimum_cut}). } 

Finally, we show how to improve upon these results further if we know
that the optimal solution is sparse. This may be a regime of interest
for certain applications where the solution space is large (e.g. structured
predictions have exponentially large candidate sets \cite{prasad2014submodular}),
and as far as we are aware , no other algorithm gives sparsity-critical
results.

\subsection{Overview of Techniques}

\label{sec:technique_overview}

In a nutshell, all are our algorithms are projected, stochastic subgradient
descent algorithms on the Lovasz extension $\hat{f}$ of a submodular
function with economical subgradient updates. The latter crucially
uses submodularity and serves as the point of departure from previous
black-box continuous optimization based methods. In this section,
we give a brief overview of our techniques.

The Lovasz extension $\hat{f}$ of a submodular function is a non-smooth
convex function whose (approximate) minimizers leads to (approximate)
SFM. Subgradient descent algorithms maintain a current iterate $x^{(t)}$and
take a step in the negative direction of a subgradient $g(x^{(t)})$
at $x^{(t)}$to get the next iterate $x^{(t+1)}$. In general, the
subgradient of a Lovasz extension takes $O(n\time)$ to compute. As
stated above, the $\Omega(n)$ lower bound on the number of iterations
needed, implies that if we naively recompute the subgradients at every
iterations, we cannot beat the quadratic barrier. Our main technical
contribution is to exploit submodularity so that $g(x^{(t+1)})$ can
be computed in sublinear time given $x^{(t)}$ and $g(x^{(t)})$.

The first implication of submodularity  is the observation (also made
by \cite{JB11,HK12}) that $\ell_{1}$-norms of the subgradients are
bounded by $O(M)$ if the submodular function is in $[-M,M]$. When
the function is integer valued,  this implies that the subgradients
are sparse and have only $O(M)$ non-zero entries. Therefore, information
theoretically, we need only $O(M)$ bits to get $g(x^{(t+1)})$ from
$g(x^{(t)})$. However, we need an algorithm to find the positions
at which they differ. To do so, we use submodularity again. We observe
that given any point $x$ and non-negative, $k$-sparse vector $e,$
the difference vector $d:=g(x+e)-g(x)$ is non-positive at points
corresponding to support of $e$ and non-negative everywhere else.
Furthemore, on a ``contiguous set'' of coordinates, the sum of these
entries in $d$ can be computed in $O(\time)$ time. Armed with this,
we create a binary search tree (BST) type data structure to find the
$O(M)$ non-zero coordinates of $d$ in $O(M\cdot\time\log n)$ time
(as opposed to $O(n\cdot\time)$ time). This, along with standard
subgradient descent analysis yields our $O(nM^{3}\time\log n)$-algorithm.

When the submodular function is real valued between $[-1,1]$, although
the $\ell_{1}$-norm is small the subgradient can have full support.
Therefore, we cannot hope to evaluate the gradient in sublinear time.
We resort to stochastic subgradient descent where one moves along
a direction whose expected value is the negative subgradient and whose
variance is bounded. Ideally, we would have liked a fast one-shot
random estimation of \textbf{$g(x^{(t+1)})$; }unfortunately we do
not how to do it. What we can do is obtain fast estimates to the difference
vector $d$ mentioned above. As discussed above, the vector $d$ has
$O(k)$ ``islands'' of non-negative entries peppered with $O(k)$
non-positive entries. We maintain a data-structure which with $O(k\time\log n)$
preprocessing time can evaluate the sums of the entries in these islands
in $O(\time\log n)$ time. Given this, we can sample a coordinate
$j\in[n]$ with probability proportional to $|d_{j}|$ in a similar
time. Thus we get a random estimate of the vector $d$ whose variance
is bounded by a constant. 

To get the stochastic subgradient, however, we need to add these difference
vectors and this accumulates the variance. To keep the variance in
control, we run the final algorithm in batches. In each batch, as
we progress we take more samples of the $d$-vector to keep the variance
in check. This however increases the sparsity (the $k$ parameter),
and one needs to balance the effects of the two. At the end of each
batch, we spend $O(n\time)$ time computing the deterministic subgradient
and start the process over. Balancing the number of iterations and
length of batches gives us the $\tilde{O}(n^{5/3}\time\eps^{-2})$-time
algorithm for $\eps$-approximate SFM.

\subsection{Related Work}

Submodularity, and indeed SFM, has a rich body of work and we refer
the reader to surveys of Fujishige \cite{fujishigeB} and McCormick\cite{Mccormick06}
for a more detailed pre-2006 version. Here we mention a few subsequent
related works which were mostly inspired by application in machine
learning. 

Motivated by applications in computer vision \cite{BVZ01,BK04} which
require fast algorithms for SFM, researchers focused on minimization
of \emph{decomposable }submodular functions which are expressible
as sum of ``simple'' submodular functions. It is assumed that simple
submodular functions can be minimized fast (either in practice or
in theory). Such a study was initiated by Stobbe and Krause \cite{SK10}
and Kolmogorov \cite{Kolmogorov12} who gave faster (than general
SFM) algorithms for such functions. More recently, motivated by work
of Jegelka et al. \cite{JBS13}, algorithms with \emph{linear }convergence\emph{
}rates \cite{NJJ14,EN15} have been obtained. That is, they get $\eps$-approximate
algorithms with dependence on $\eps$ being $\log(1/\eps)$. \textbf{}.

We end our introductory discussion by mentioning the complexity of
\emph{constrained }SFM where one wishes to minimize over sets satisfying
some constraints. In general constrained SFM is much harder than unconstrained
SFM. For instance the minimum cut problem with cardinality constraints
becomes the balanced partitioning problem which is APX-hard. More
generally, Svitkina and Fleischer \cite{SF11} show that a large class
of constrained SFM problems cannot be approximated to better than
$\tilde{O}(\sqrt{n})$ factors without making exponentially many queries.
In contrast, Goemans and Soto \cite{GS13} prove that symmetric submodular
functions can be minimized over a large class of constraints. Inspired
by machine learning applications, Iyer et al. \cite{IJB13,IB13} give
algorithms for a large class of constrained SFM problems which have
good approximation guarantees if the \emph{curvature} of the functions
are small.

\paragraph{}

%% file: preliminaries.tex
\section{Preliminaries}

Here we introduce notations and general concepts used throughout this
paper. 

\subsection{General Notation \label{sec:prelim:notation}}

We let $[n]\defeq\{1,...,n\}$ and $[0,1]^{n}\defeq\{x\in\R^{n}\,:\,x_{i}\in[0,1]\,\,\,\forall i\in[n]\}$.
Given a permutation $P=(P_{1},...,P_{n})$ of $[n]$, let $P[j]\defeq\{P_{1},P_{2},...,P_{j}\}$
be the set containing the first $j$ elements of $P$. Any point $x\in\R^{n}$
defines the permutation \emph{$P_{x}$ consistent with $x$ }where
$x_{P_{1}}\geq x_{P_{2}}\geq...\geq x_{P_{n}}$with ties broken lexicographically.
We denote by $\mathbf{1}_{i}\in\R^{n}$ the indicator vector for coordinate
$i$, i.e. $\mathbf{\mathbf{1}}_{i}$ has a $1$ in coordinate $i$
and a $0$ in all other coordinates. We call a vector $s$-\textit{sparse}
if it has at most $s$ non-zero entries.

\subsection{Submodular Functions \label{sec:prelim:submodular_functions}}

Throughout this paper $f\,:\,2^{U}\rightarrow\R$ denotes a submodular
function on a ground set $U$. For notational convenience we assume
without loss of generality that $U=[n]$ for some positive integer
$n$ and that $f(\emptyset)=0$ (as this can be enforced by subtracting
$f(\emptyset)$ from for the value of all sets while preserving submodularity).
Recall that $f$ is submodular if and only if it obeys the property
of diminishing marginal returns: for all $S\subseteq T\subseteq[n]$
and $i\notin T$ we have 
\[
f(S\cup\{i\})-f(S)\geq f(T\cup\{i\})-f(T)\,.
\]
We let $\OPT\defeq\min_{S\subseteq[n]}f(S)$ be the minimum value
of $f$. We denote by $\runtime$ the time it takes to evaluate $f$
on a set $S$. More precisely, we assume given a linked list storing
a permutation $P$ of $[n]$, and a position $k$, we can evaluate
$f(P[k])$ in $\time$ time. 

\subsection{The Lovasz Extension \label{sec:prelim:lovasz_extension}}

Our results make extensive use of the Lovasz extension, a convex,
continuous extension of a submodular function to the interior of the
$n$-dimensional hypercube, i.e. $[0,1]^{n}$.
\begin{defn}[Lovasz Extension]
 Given a submodular function $f$, the Lovasz extension of $f$,
denoted as $\hat{f}\,:\,[0,1]^{n}\rightarrow\R$, is defined for all
$x\in[0,1]^{n}$ by $\hat{f}(x)=\sum_{j\in[n]}(f([P[j])-f(P[j-1]))x_{i_{j}}$
where $P=P_{x}=(P_{1},...,P_{n})$ is the permutation consistent with
$x$.
\end{defn}
Note that since $f(\emptyset)=0$ this definition is equivalent to
\begin{equation}
\hat{f}(x)=f(P[n])x_{P_{n}}+\sum_{j\in[n-1]}f([P[j])(x_{P_{j}}-x_{P_{j+1}})\,.\label{eq:lovasz_non_negative}
\end{equation}
{} We make use of the following well known facts regarding submodular
functions (see e.g. \cite{Lovasz83,fujishigeB}).
\begin{thm}[Lovasz Extension Properties]
\label{thm:lovasz_extension_properties}  The following are true
for all $x\in[0,1]^{n}$:%
\end{thm}
\begin{itemize}
\item \setlength{\itemsep}{-1mm}\textbf{Convexity}: The Lovasz extension
is convex.
\item \textbf{Consistency}: For $x\in\{0,1\}^{n}$ we have $\hat{f}(x)=f(S(x))$
where $S(x)=\{i\in S\,:\,x_{i}=1\}$.
\item \textbf{Minimizers}: $\min_{x\in[0,1]^{n}}\hat{f}(x)=\min_{S\subseteq[n]}f(S)$.
\item \textbf{Subgradients}: The vector $g(x)\in\R^{n}$ defined by $g(x)_{P_{k}}\defeq f(P[k])-f(P[k-1])$
is a subgradient of $\hat{f}$ at $x$, where $P=P_{x}$ is the permutation
consistent with $x$. Let us call this the Lovasz subgradient. %

\end{itemize}
We conclude with a few straightforward computational observations
regarding the Lovasz extension and its subgradients. First note that
for $x\in[0,1]^{n}$ we can evaluate $\hat{f}(x)$ or compute $g(x)$
in time $O(n\time+n\log n)$ simply by sorting the coordinates of
$f$ and evaluating $f$ at the $n$ desired sets. Also, note that
by (\ref{eq:lovasz_non_negative}) the Lovasz extension evaluated
at $x\in[0,1]^{n}$ is a non-negative combination of the value of
$f$ at $n$ sets. Therefore computing the smallest of these sets
gives a set $S\subseteq[n]$ such that $f(S)\leq\hat{f}(x)$ and we
can clearly compute this, again in $O(n\time+n\log n)$ time. Therefore
for any algorithm which approximately minimizes the Lovasz extension
with some (additive) error $\eps$, we can always find a set $S$
achieving the same error on $f$ by just paying an additive $O(n\time+n\log n)$
in the running time.

\subsection{Subgradient Descent \label{sec:prelim:subgradient_descent}}

Our algorithmic results make extensive use of subgradient descent
(or mirror descent) and their stochastic analogs. Recall that for
a convex function $h\,:\,\chi\rightarrow\R$, where $\chi\subseteq\R^{n}$
is a compact convex set, a vector $g\in\R^{n}$ is a \emph{subgradient}
of $h$ at $x\in\chi$ if for all $y\in\chi$ we have 
\[
h(y)\geq h(x)+g^{\top}(y-x)\,.
\]
For such an $h$ we let $\partial h(x)$ denote the set of subgradients
of $h$ at $x$. An algorithm that on input $x$ outputs $\tilde{g}(x)\in\partial h(x)$
is a \emph{subgradient oracle} for $h$. Similarly, an algorithm that
on input $x$ outputs a \emph{random} $\tilde{g}(x)$ such that $\E\tilde{g}(x)\in\partial h(x)$
is a \emph{stochastic subgradient oracle} for $h$. 

One of our main algorithmic tools is the well known fact that given
a (stochastic) subgradient oracle we can minimize a convex function
$h$. Such algorithms are called\emph{ (stochastic) subgradient descent}
algorithms and fall into a more general framework of algorithms known
as \emph{mirror descent}. These algorithms are very well studied and
there is a rich literature on the topic. Below we provide one specific
form of these algorithms adapted from \cite{Bubeck15} that suffices
for our purposes.
\begin{thm}[Projected (Stochastic) Subgradient Descent\footnote{This is Theorem~6.1 from \cite{Bubeck15} restated where we used
the ``ball setup'' with $\Phi(x)=\frac{1}{2}\norm x_{2}^{2}$ so
that $\dist=\R^{n}$ and $D_{\Phi}(x,y)=\frac{1}{2}\norm{x-y}_{2}^{2}$.
We also used that $\argmin_{x\in\chi}\eta g^{\top}x+\frac{1}{2}\norm{x-x_{t}}_{2}^{2}=\argmin_{x\chi}\norm{x-(x_{t}-\eta g)}_{2}^{2}$.}]
\label{thm:subgradient_descent} Let $\chi\subseteq\R^{n}$ denote
a compact convex set, $h\,:\,\chi\rightarrow\R$ be a convex function,
$\tilde{g}$ be a (stochastic) subgradient oracle for which $\E\norm{\tilde{g}(x)}_{2}^{2}\leq B^{2}$
for all $x\in\chi$, and $R^{2}\defeq\sup_{x\in\chi}\frac{1}{2}\norm x_{2}^{2}$
. Now consider the iterative algorithm starting with 
\[
x^{(1)}:=\argmin_{x\in\chi}\norm x_{2}^{2}
\]
and for all $s$ we compute 
\[
x^{(s+1)}:=\argmin_{x\in\chi}\norm{x-(x^{(s)}-\eta\tilde{g}(x^{(s)}))}_{2}^{2}
\]
Then for $\eta=\frac{R}{B}\sqrt{\frac{2}{t}}$ we have
\[
\E h\left(\frac{1}{t}\sum_{i\in[t]}x^{(s)}\right)-\min_{x\in\chi}h(x)\leq RB\sqrt{\frac{2}{t}}\,.
\]
We refer to this algorithm as projected stochastic subgradient descent
when $\tilde{g}$ is stochastic and as projected subgradient descent
when $\tilde{g}$ is deterministic, though we often omit the term
\emph{projected }for brevity. Note that when $\tilde{g}$ is deterministic
the results are achieved exactly rather than in expectation. \end{thm}

%% file: pseduopolynomial.tex
\section{Faster Submodular Function Minimization\label{sec:fast_sfm}}

In this section we provide faster algorithms for SFM. In particular
we provide the first nearly linear time pseudopolynomial algorithm
for SFM and the first subquadratic additive approximation algorithm
for SFM. Furthermore, we show how to obtain even faster running times
when the SFM instance is known to have a sparse solution. 

All our algorithms follow the same broad algorithmic framework of
using subgradient descent with a specialized subgradient oracle. Where
they differ is in how the structure of the submodular functions is
exploited in implementing these oracles. The remainder of this section
is structured as follows: in Section\textbf{~\ref{sec:fast_sfm:framework}}
we provide the algorithmic framework we use for SFM, in Section\textbf{~\ref{sec:fast_sfm:subgrad}},
we prove structural properties of submodular functions that we use
to compute subgradients, in Section\textbf{~\ref{sec:fast_sfm:pseudopoly}},
we describe our nearly linear time pseodopolynomial algorithms, in
Section\textbf{~\ref{sec:fast_sfm:subquad}}, we describe our subquadratic
additive approximation algorithm, and in Section\textbf{~\ref{sec:fast_sfm:sparse}},
we show how to improve these results when SFM has a sparse solution. 

We make minimal effort to control logarithmic factors in through this
section and note that some of the  factors come from sorting and therefore
maybe can be removed depending on the desired computational model.

\subsection{Algorithmic Framework \label{sec:fast_sfm:framework}}

All our algorithms for SFM follow the same broad algorithmic framework.
We consider the Lovasz extension $\hat{f}:[0,1]^{n}\rightarrow\R$,
and perform projected (stochastic) subgradient descent on $\hat{f}$
over the convex domain $\chi=[0,1]^{n}$. While the subgradient oracle
construction differs between algorithms (and additional care is used
to improve when the solution is sparse, i.e. Section~\ref{sec:fast_sfm:sparse})
the rest of algorithms for Section~\ref{sec:fast_sfm:pseudopoly}
and Section~\ref{sec:fast_sfm:subquad} are identical.

In the following, Lemma~\ref{lem:framework}, we encapsulate this
framework, bounding the performance of projected (stochastic) subgradient
descent to the Lovasz extension, i.e. applying Theorem~\ref{thm:subgradient_descent}
to $\hat{f}$ over $\chi=[0,1]^{n}$. Formally, we abstract away the
properties of a subgradient oracle data structure that we need to
achieve a fast algorithm. With this lemma in place the remainder of
the work in Section~\ref{sec:fast_sfm:subgrad}, Section~\ref{sec:fast_sfm:pseudopoly},
and Section~\ref{sec:fast_sfm:subquad} is to show how to efficiently
implement the subgradient oracle in the particular setting.

\global\long\def\Tg{\mathrm{T_{g}}}

\begin{lem}
\label{lem:framework} Suppose that there exists a procedure which
maintains $(x^{(i)},\tilde{g}^{(i)})$ satisfying the invariants:
(a) $\tilde{g}^{(i)}$ is $k$-sparse, (b) $\E[\tilde{g}^{(i)}]=g(x^{(i)})$
is the Lovasz subgradient at $x^{(i)},$(c) $\E\norm{\tilde{g}^{(i)}}_{2}^{2}\leq B^{2}$.
Furthermore, suppose given any $e^{(i)}$ which is $k$-sparse, the
procedure can update to $(x^{(i+1)}=x^{(i)}+e^{(i)},\tilde{g}^{(i+1)})$
in time $\Tg$. Then, for any $\eps>0,$ we can compute a set $S$
with $\E[f(S)]\le\OPT+\eps$ in time $O(nB^{2}\eps^{-2}\Tg+n\time+n\log n)$.
If invariants (b) and (c) hold without expectation, then so does our
algorithm.

\end{lem}
\begin{proof}
We invoke Theorem~\ref{thm:subgradient_descent} on the convex function
$\hat{f}\,:\,[0,1]^{n}\rightarrow\R$ over the convex domain $\chi=[0,1]^{n}$
to obtain the iterates where we use the given subgradient oracle.
Clearly 
\[
x^{(1)}=\argmin_{x\in[0,1]^{n}}\frac{1}{2}\norm x_{2}^{2}=0\in\R^{n}
\]
and
\[
R^{2}=\sup_{x\in[0,1]^{n}}\frac{1}{2}\norm x_{2}^{2}=\frac{1}{2}\norm 1_{2}^{2}=\frac{n}{2}.
\]
Consequently, as long as we implement the projection step for $T=O(nB^{2}\eps^{-2})$
steps (each step requiring $\Tg$ time), then Theorem~\ref{thm:subgradient_descent}
yields
\[
\E\hat{f}\left(\frac{1}{T}\sum_{i\in[T]}x^{(i)}\right)-\min_{x\in\chi}\hat{f}(x)\leq RB\sqrt{\frac{2}{T}}\leq\sqrt{\frac{nB^{2}}{T}}\leq\eps\,.
\]

Furthermore, as we argued in Section~\ref{sec:prelim:lovasz_extension}
we can compute $S$ with $f(S)\leq\hat{f}(\frac{1}{T}\sum_{i\in[T]}x^{(i)})$
in the time it takes to compute $\frac{1}{T}\sum_{i\in[T]}x^{(i)}$
plus additional $O(n\runtime+n\log n)$ time. To prove the lemma all
that remains to to reason about the complexity of computing the projection,
i.e. $x^{(t+1)}$, given that all the subgradients we compute are
$s$-sparse. However, since $x^{(t+1)}=\argmin_{x\in[0,1]^{n}}\norm{x-(x^{(t)}-\eta\tilde{g}(x^{(t)})}_{2}^{2}$
decouples coordinate-wise -- note that $x^{(t+1)}=\mathrm{median}\{0\,,\,x^{(t)}-\eta\tilde{g}(x^{(t)})\,,\,1\}$,
we subtract $\eta\tilde{g}(x^{(t)})$ from $x^{(t)}$ and if any coordinate
is less than $0$ we set it to $0$ and if any coordinate is larger
than $1$ we set it to $1$. Thus the edit vector $e^{(i)}$ is of
sparsity $\leq k$. Combining these facts yields the described running
time.
\end{proof}

\subsection{Subgradients of the Lovasz Extension\label{sec:fast_sfm:subgrad}}

Here we provide structural results of submodular function that we
leverage to compute subgradients of submodular functions in $o(n)$
time on average. %
{} First, in Lemma~\ref{lem:bounded-subgradients} we state a result
due to Jegelka and Bilmes \cite{JB11}(also Hazan and Kale \cite{HK12})
which puts an upper bound on the $\ell_{1}$ norm of subgradients
of the Lovasz extension provided we have an upper bound on the maximum
absolute value of the function. We provide a short proof for completeness.
{} %
{} 
\begin{lem}[Subgradient Upper Bound]
\label{lem:bounded-subgradients}  If $|f(S)|\leq M$ for all $S\subseteq[n]$,
then $\norm{g(x)}_{1}\leq3M$ for all $x\in[0,1]^{n}$ and for all
subgradients $g$ of the Lovasz extension.\end{lem}
\begin{proof}
For notational simplicity suppose without loss of generality (by changing
the name of the coordinates) that $P(x)=(1,2,...,n)$, i.e. $x_{1}\geq x_{2}\geq...\geq x_{n}$.
Therefore, for any $i\in[n],$we have $g_{i}=f([i])-f([i-1])$. Let
$r_{1}\leq r_{2}\leq...,\leq r_{R}$ denote all the coordinates such
that $g_{r_{i}}>0$ and let $s_{1}\leq s_{2}\leq...\leq s_{S}$ denote
all the coordinates such that $g_{s_{i}}<0$. 

We begin by bounding the contribution of the positive coordinate,
the $g_{r_{i}}$, to the norm of the gradient, $\norm g_{1}$. For
all $k\in[R]$ let $R_{k}\defeq\left\{ r_{1},...,r_{k}\right\} $
with $R_{0}=\emptyset$. By assumption we know that that $f(R_{0})=\emptyset$.
Furthermore, by submodularity, i.e. diminishing marginal returns,
we know that for all $i\in[R]$
\[
f(R_{i})-f(R_{i-1})\geq f([r_{i}])-f([r_{i}-1])=:g_{r_{i}}=\left|g_{r_{i}}\right|
\]
Consequently $f(R_{R})-f(R_{0})=\sum_{i\in[R]}f(R_{i})-f(R_{i-1})\geq\sum_{i\in[R]}\left|g_{r_{i}}\right|$.
Since $f(R_{0})=0$ and $f(R_{R})\leq M$ by assumption we have that
$\sum_{i\in[k]}\left|g_{r_{i}}\right|\leq M\,.$

Next, we bound the contribution of the negative coordinatess, the
$g_{s_{i}}$, similarly. For all $k\in[S]$ let $S_{k}\defeq\left\{ s_{1},...,s_{k}\right\} $
with $S_{0}=\emptyset$. By assumption we know that that $f(S_{0})=\emptyset$.
Define $V:=[n]\setminus S$. Note that for all $i\in[S],$ the set
$V\cup S_{i-1}$ is a superset of $[s_{i}-1].$ Therefore, submodularity
gives us for all $i\in[S]$, 

\[
f(V\cup S_{i})-f(V\cup S_{i-1})\leq f([s_{i}])-f([s_{i}-1])=g_{s_{i}}=-\abs{g_{s_{i}}}
\]

Summing over all $i,$ we get $f([n])-f(V)\leq\sum_{i\in[S]}-\left|g_{s_{i}}\right|$.
Since $f([n])\geq-M$ and $f(V)\leq M$ we have that $\sum_{i\in[S]}\left|g_{s_{i}}\right|\leq2M\,.$
Combining these yields that $\norm g_{1}=\sum_{i\in[n]}\left|g_{i}\right|=\sum_{i\in[R]}\left|g_{r_{i}}\right|+\sum_{i\in[S]}\left|g_{s_{i}}\right|\leq3M\,.$
\end{proof}
Next, in Lemma~\ref{lem:subgrad_monotonicity} we provide a simple
but crucial monotonicity property of the subgradient of $\hat{f}$.
In particular we show that if we add (or remove) a positive vector
from $x\in[0,1]^{n}$ to obtain $y\in[0,1]^{n}$ then the gradients
of the \emph{untouched} coordinates all decrease (or increase).
\begin{lem}[Subgradient Monotonicity]
\label{lem:subgrad_monotonicity}  Let $x\in[0,1]^{n}$ and let $d\in\R_{\geq0}^{n}$
be such that $y=x+d$ (resp. $y=x-d$). Let $S$ denote the non-zero
coordinates of $d$. Then for all $i\notin S$ we have $g(x)_{i}\geq g(y)_{i}$
(resp. $g(x)_{i}\leq g(y)_{i}$).\end{lem}
\begin{proof}
We only prove the case of $y=x+d$ as the proof of the $y=x-d$ case
is analagous. Let $P^{(x)}$and $P^{(y)}$ be the permutations consistent
with $x$ and $y$. Note that $P^{(y)}$ can be obtained from $P^{(x)}$
by moving a subset of elements in $S$ to the left, and the relative
ordering of elements \emph{not in }$S$ remains the same. Therefore,
for any $i\notin S,$ if $r$ is its rank in $P^{(x)},$ that is,
$P_{r}^{(x)}=i,$ and $r'$is its rank in $P^{(y)}$ , then we must
have $P^{(y)}[r']\supseteq P^{(x)}[r].$ By submodularity, $g(y)_{i}=f(P^{(y)}[r'])-f(P^{(y)}[r'-1])\leq f(P^{(x)}[r])-f(P^{(x)}[r-1]))=g(x)_{i}.$
\end{proof}
Lastly, we provide Lemma~\ref{lem:subgrad_interval} giving a simple
formula for the sum of multiple coordinates in the subgradient. 

\begin{lem}[Subgradient Intervals]
\label{lem:subgrad_interval}  Let $x\in[0,1]^{n}$ and let $P$
be the permutation consistent with $x$. For any positive integers
$a\leq b$, we have $\sum_{i=a}^{b}g(x)_{P_{i}}=f(P[b])-f(P[a-1])$.\end{lem}
\begin{proof}
This follows immediately from the definition of $g(x)$: $\sum_{i=a}^{b}g(x)_{P_{i}}=\sum_{i=a}^{b}(f(P[i])-f(P[i-1]))=f(P[b])+\sum_{i=a}^{b-1}f(P[i])-\sum_{i=a}^{b-1}f(P[i])-f(P[a-1])\,.$
\end{proof}

\subsection{Nearly Linear in $n$, Pseudopolynomial Time Algorithm \label{sec:fast_sfm:pseudopoly}}

Here we provide the first nearly linear time pseudopolynomial algorithm
for submodular function minimization. Throughout this section we assume
that our submodular function $f$ is integer valued with $|f(S)|\leq M$
for all $S\subseteq[n]$. Our goal is to deterministically produce
an exact minimizer of $f$. The primary result of this section is
showing the following, that we can achieve this in $\otilde(nM^{3}\runtime)$
time:
\begin{thm}
\label{thm:runtime_pseudopoly} Given an integer valued submodular
function $f$ with $|f(S)|\leq M$ for all $S\subseteq[n]$ in time
$O(nM^{3}\time\log n)$ we can compute a minimizer of $f$.
\end{thm}
We prove the theorem by describing $(x^{(i)},\tilde{g}(x^{(i)})$)
in Lemma~\ref{lem:framework}. In fact, in this case, $\tilde{g}$
will \emph{deterministically }be the subgradient of the Lovasz extension.
In Lemma~\ref{lem:sparse_subgradients}, we prove that the Lovasz
subgradient is $O(M)$-sparse and so $\|g\|_{2}^{2}\leq O(M^{2})$.
Thus, Conditions (a), (b), and (c) are satisfied with $B^{2}=O(M^{2})$.
The main contribution of this section is Lemma \ref{lem:subgradient-update},
where we show that $\Tg=O(M\log n\cdot\time)$, that is the subgradient
can be updated in this much time. A pseudocode of the full algorithm
can be found in Section~\ref{alg:exact-1}.

\begin{lem}
\label{lem:sparse_subgradients} For integer valued $f$ with $|f(S)|\leq M$
for all $S$ the subgradient $g(x)$ has at most $3M$ non-zero entries
for all $x\in[0,1]^{n}$.\end{lem}
\begin{proof}
By Lemma~\ref{lem:bounded-subgradients} we know that $\norm{g(x)}_{1}\leq3M$.
However, since $g(x)_{P_{i}}=f(P[i])-f(P[i-1])$ and since $f$ is
integer valued, we know that either $g(x)_{P_{i}}=0$ or $\left|g(x)_{P_{i}}\right|\geq1$.
Consequently, there are at most $3M$ values of $i$ for which $g(x)_{i}\neq0$.\end{proof}
\begin{lem}
\label{lem:subgradient-update} With $O(n\cdot\time)$ preprocessing
time the following data structure can be maintained. Initially, one
is input $x^{(0)}\in[0,1]^{n}$ and $g(x_{0})$. Henceforth, for all
$i,$ given $g(x^{(i)})$ and a vector $e^{(i)}$ which is $k$-sparse,
in $O(k\log n+k\time+M\time\log n)$ time one can update $g(x^{(i)})$
to the gradient $g(x^{(i+1)})$ for $x^{(i+1)}=x^{(i)}+e^{(i)}$.\end{lem}
\begin{proof}
The main idea is the following. Suppose $e^{(i)}$ is non-negative
(we later show how to easily reduce to the case where all coordinates
in $e^{(i)}$ have the same sign and the non-positive case is similar)
Thus, by Lemma~\ref{lem:subgrad_monotonicity}, for all coordinates
not in support of $e^{(i)},$ the gradient goes down. Due to Lemma
\ref{lem:sparse_subgradients}, the total number of change is $O(M)$,
and since we can evaluate the sum of gradients on intervals by Lemma
\ref{lem:subgrad_interval}, a binary search procedure allows us to
find \emph{all }the gradient changes in $O(M\log n\cdot\time)$ time.
We now give full details of this idea.

We store the coordinates of $x^{(i)}$ in a balanced binary search
tree (BST) with a node for each $j\in[n]$ keyed by the value of $x_{j}^{(i)}$;
ties are broken consistently, e.g. by using the actual value of $j$.
We take the order of the nodes $j\in[n]$ in the binary search tree
to define the permutation $P^{(i)}$ which we also store explicitly
in a link-list, so we can evaluate $f(P^{(i)}[k])$ in $O(\time)$
time for any $k$.\textbf{ }Note that each node of the BST corresponds
to a subinterval of $P^{(i)}$ given by the children of that node
in the tree. At each node of the BST, we store the sum of $g(x^{(i)})_{j}$
for all children $j$ of that node, and call it the \emph{value }of
the node. Note by Lemma \ref{lem:subgrad_interval} each individual
such sum can be computed with 2 calls to the evaluation oracle. Finally,
in a linked list, we keep all indices $j$ such that $g(x^{(i)})_{j}$
is non-zero and we keep pointers to them from their corresponding
node in the binary search tree. Using the binary search tree and the
linked list, one can clearly output the subgradient. Also, given $x^{(0)}$
, in $O(n\cdot\time)$ time one can obtain the initialization. What
remains is to describe the update procedure. 

We may assume that all non-zero entries of $e^{(i)}$ are the same
sign; otherwise write $e^{(i)}:=e_{+}^{(i)}+e_{-}^{(i)}$, and perform
two updates. WLOG, lets assume the sign is + (the other case is analogous).
Let $S$ be the indices of $e^{(i)}$ which are non-zero. 

First, we change the key for each $j\in[n]$ such that $x_{j}^{(i+1)}\neq x_{j}^{(i)}$
and update the BST. Since we chose a consistent tie breaking rule
for keying, only these elements $j\in[n]$ will change position in
the permutation $P^{(i+1)}$. Furthermore, performing this update
while maintaining the subtree labels can be done in $O(k\log n)$
time as it is easy to see how to implement binary search trees that
maintain the subtree values even under rebalancing. For the time being,
we retain the old values as is. 

For brevity, let $g^{(i)}$ and $g^{(i+1)}$ denote the gradients
$g(x^{(i)})$ and $g(x^{(i+1)})$, respectively. Since we assume all
non-zero changes in $e^{(i)}$ are positive, by Lemma \ref{lem:subgrad_monotonicity},
we know that $g_{j}^{(i+1)}\le g_{j}^{(i)}$ for all $j\notin S$.
First, since $|S|\leq k$, for all $j\in S,$ we go ahead and compute
$g_{j}^{(i+1)}$ in $O(k\time)$ time. For each such $j$ we update
the value of the nodes from $j$ to the root, by adding the difference
$(g_{j}^{(i+1)}-g_{j}^{(i)})$ to each of them. Next, we perform the
following operation top-down start at the root: at each node we compare
the current subtree value stored at this node with what the value
actually should be with $g^{(i+1)}$ . Note that since we know $P^{(i+1)}$,
the latter can be computed with 2 evaluation queries. The simple but
crucial observation is that if at any node $j$ these two values match,
then we are guaranteed that $g_{k}^{(i+1)}=g_{k}^{(i)}$ for all $k$
in the tree rooted at $j$ and we do not need to recurse on the children
of this node. The reason for equality is that for all the children,
we must have $g_{k}^{(i+1)}\leq g_{k}^{(i)}$ by Lemma \ref{lem:subgrad_monotonicity}
, and so if the sum is equal then we must have equality everywhere.
Since there are at most $O(M)$ coordinates change, this takes $O(M\time\log n)$
for updating \emph{all }the changes to $g^{(i+1)}$ for the binary
search tree. During the whole process, whenever a node changes from
non-zero to zero or from zero to non-zero, we can update the linked-list
accordingly. 
\end{proof}

\begin{proof}[Proof of Theorem~\ref{thm:runtime_pseudopoly}]
 We apply Lemma~\ref{lem:framework} giving the precise requirements
of our subgradient oracle. We know that the subgradients we produce
are always $O(M)$ sparse by Lemma~\ref{lem:sparse_subgradients}
and satisfy $B^{2}=O(M^{2})$. Consequently, we can simply instantiate
Lemma~\ref{lem:subgradient-update} with $k=O(M)$ to obtain our
algorithm. Furthermore, since $f$ is integral we know that so long
as we have a set additive error less than $1$, i.e. $\epsilon<1$,
the set is a minimizer. Consequently, we can minimize in the time
given by the cost of adding the cost of Lemma~2, with the Lemma~\ref{lem:sparse_subgradients}
initialization cost, plus the Lemma~\ref{lem:sparse_subgradients}
cost for $T=O(nM^{2})$ iterations, yielding 
\[
O\left(n(\runtime+\log n+M^{3})+n+M\runtime+(M\log n+M\runtime\log n)\cdot nM^{2}\right)=O(nM^{3}\time\log n)\,.
\]

\end{proof}

\subsection{Subquadratic Additive Approximation Algorithm \label{sec:fast_sfm:subquad}}

Here we provide the first subquadratic  additive approximation algorithm
for submodular function minimization. Throughout this section we assume
that $f$ is real valued with $|f(S)|\leq1$ for all $S\subseteq[n]$.
Our goal is to provide a randomized algorithm that produces a set
$S\subseteq[n]$ such $\E f(S)\leq\OPT+\epsilon$. The primary result
of this section is showing the following, that we can achieve this
in $O(n^{5/3}\epsilon^{-2}\log^{4}n)$ time:
\begin{thm}
\label{thm:runtime_additive} Given a submodular function $f:2^{[n]}\rightarrow\R$
with $|f(S)|\leq1$ for all $S\subseteq V$, and any $\eps>0$, we
we can compute a random set $S$ such that $\E f(S)\leq\OPT+\epsilon$
in time $O(n^{5/3}\epsilon^{-2}\time\log^{4}n)$.
\end{thm}
The proof of this theorem has two parts. Note that the difficulty
in the real-valued case is that we can no longer assume the Lovasz
gradients are sparse, and so we cannot do naive updates. Instead,
we use the fact that the gradient has small $\ell_{2}$norm to get
sparse \emph{estimates }of the gradient. This is the first part where
we describe a sampling procedure which given any point $x$ and a
$k$-sparse vector $e$, returns a good and sparse estimate to the
\emph{difference }between the Lovasz gradient at $x+e$ and $x$.
The second issue we need to deal with is that if we naively keep using
this estimator, then the error (variance) starts to accumulate. The
second part then shows how to use the sampling procedure in a ``batched
manner'' so as to keep the total variance under control, restarting
the whole procedure with a certain frequency. A pseudocode of the
full algorithm can be found in Section \ref{alg:exact-1}.
\begin{lem}
\label{lem:subgradient-update_additive-1} Suppose a vector $x\in[0,1]^{n}$
is stored in a BST sorted by value. Given a $k$-sparse vector $e$
which is either non-negative or non-positive, and an integer $\ell\geq1$,
there is a randomized sampling procedure which returns a vector $z$
with the following properties: (a) $\E[z]=g(x+e)-g(x)$, (b) $\E[\norm{z-\E[z]}_{2}^{2}]=O(1/\ell)$,
and (c) the number of non-zero coordinates of $z$ is $O(\ell)$.
The time taken by the procedure is $O((k+\ell)\cdot\time\log^{2}n)$.\end{lem}
\begin{proof}
We assume that each non-zero value of $e$ is positive as the other
case is analogous. Note that $x$ is stored in a BST  , and the permutation
$P_{x}$ consistent with $x$ is stored in a doubly linked list. Let
$S$ be the set of positive coordinates of $e$ with $|S|=k$ and
let $y$ denote the vector $x+e$. We compute $P_{y}$ in $O(k\log n)$
time. 

Let $I_{1},\ldots,I_{2k}\subseteq[n]$ denote the subsets of the coordinates
that correspond to the intervals which are contiguous in both $P_{x}$
and $P_{y}$. Note that these are $\leq2k\mbox{\,such intervals, and some of them can be empty.}$
We store the pointers to the endpoints of each interval in the BST.
This can be done in $O(k\log n)$ time as follows. First compute the
coarse intervals which are contiguous in $P_{x}$ in $O(k)$ time.
These intervals will be refined when we obtain $P_{y}$. In $O(k\log n)$
time, update the BST so that for every node we can figure out which
coarse interval it lies in $O(\log n)$ time. This is done by walking
up the BST for every end point of all the $k$ intervals and storing
which ``side'' of the interval they lie in. Given a query node,
we can figure out which interval it lies in by walking up the BST
to the root. Finally, for all nodes in $S$, when we update the BST
in order to obtain $P_{y}$, using the updated data structure in $O(\log n)$
time figure out which coarse interval it lies in and refine that interval. 

For each $j\in S,$ we compute $d_{j}\defeq g(y)_{j}-g(x)_{j}$ explicitly.
This can be done in $O(k\time)$time using $P_{x}$and $P_{y}$.For
$r\in[2k]$, we define $D_{r}:=\sum_{j\in I_{r}}(g(y)_{j}-g(x)_{j})$.
Since each $I_{r}$ is a contiguous interval in both $P_{x}$ and
$P_{y},$ Lemma~\ref{lem:subgrad_interval} implies that we can store
all $D_{r}$ in $O(k\cdot\time)$ time in look-up tables. Note that
by monotonicity Lemma~\ref{lem:subgrad_monotonicity} each summand
in $D_{r}$ is of the same sign, and therefore summing the absolute
values of $D_{r}$'s and $d_{j}$'s gives $\norm{g(y)-g(x)}_{1}.$
We store this value of the $\ell_{1}$ norm.

Now we can state the randomized algorithm which returns the vector
$z.$ We start by sampling either a coordinate $j\in S$ with probability
proportional to $|d_{j}|$ , or an interval $I_{r}$ with probability
proportional to $\abs{D_{r}}$. If we sample an interval, then iteratively
sample sub-intervals $I'\subset I_{r}$ proportional to  $\sum_{j\in I'}(g(y)_{j}-g(x)_{j})$
till we reach a single coordinate $j\notin S.$ Note that any $j\in[n]$
is sampled with probability proportional to $\abs{g(y)_{j}-g(x)_{j}}$. 

We now show how to do this iterative sampling in $O((\time+\log n)\log n)$
time. Given $I_{r}$, we start from the root of the BST and find a
node closest to the root which lies in $I_{r}.$ More precisely, since
for every ancestor of the endpoints of $I_{r}$, if it doesn't belong
to the interval we store which ``side'' of the tree $I_{r}$ lies
in, one can start from the root and walk down to get to a node inside
$I_{r}$. This partitions $I_{r}$into two subintervals and we randomly
select $I'$proportional to $\sum_{j\in I'}(g(y)_{j}-g(x)_{j})$ .
Since sub-intervals are contiguous in $P_{y}$ and $P_{x,}$ this
is done in $O(\time)$ time. We then update the information at every
ancestor node of the endpoints of the sampled $I'$ in $O(\log n)$
time. Since each iteration decreases the height of the least common
ancestor of the endpoints of $I'$, in $O(\log n)$ iterations (that
is the height of the tree), we will sample a singleton $j\not\notin S$. 

In summary, we can sample $j\in[n]$ with probability proportional
to $g(y)_{j}-g(x)_{j}$ in $O((\time+\log n)\log n)$ time. If we
sample $j,$ we return the (random) vector 
\[
z:=\norm{g(y)-g(x)}_{1}\cdot\sign(g(y)_{j}-g(x)_{j})\cdot\mathbf{1}_{j}
\]
where recall $\mathbf{1}_{j}$ is the vector with 1 in the $j$th
coordinate and zero everywhere else. Note that given $j,$ computing
$z$ takes $O(\time+\log n)$ time since we have to evaluate $g(y)_{j}$and
$g(x)_{j}.$ Recall, we already know the $\ell_{1}$norm. Also note
by construction, $\E[z]$ is precisely the vector $g(y)-g(x).$ To
upper bound the variance, note that

\[
\E[\norm{z-\E z}_{2}^{2}]\leq\E[\norm z_{2}^{2}]=\norm{g(y)-g(x)}_{1}^{2}\leq9\cdot\max_{S\subseteq V}|f(S)|\leq9
\]
by Lemma~\ref{lem:bounded-subgradients} and the fact that $|f(S)|\leq1$.
Also observe that $z$ is $1$-sparse. 

Given $\ell,$ we sample independently $\ell$ such random $z$'s
and return their \emph{average. }The expectation remains the same,
but the variance  scales down by $\ell.$ The sparsity is at most
$\ell.$The total running time is $O(k(\time+\log n)+\ell(\time+\log n)\log n)$.
This completes the proof of the lemma.

\end{proof}
We now complete the proof of \emph{Theorem \ref{thm:runtime_additive}.}
\begin{proof}
(\emph{Theorem \ref{thm:runtime_additive}}) The algorithm runs in
batches (as mentioned before, the pseudocode is in Section \ref{alg:exact-1}.)
At the beginning of each batch, we have our current vector $x^{(0)}$
as usual stored in a BST. We also compute the Lovasz gradient $g^{(0)}=g(x^{(0)})$
spending $O(n\log n\time)$ time. The batch runs for $T=\Theta(n^{1/3})$
steps. At each step $t\in[T],$ we need to specify an estimate $\tilde{g}^{(t)}$
to run the (stochastic) subgradient procedure as discussed in Lemma
\ref{lem:framework}. For $t=0$, since we know $g^{(0)}$ explicitly,
we get $\tilde{g}^{(0)}$ by returning $\norm{g^{(0)}}_{1}\mathrm{sign}(g_{j}^{(0)})\mathbf{1_{j}}$
with probability proportional to $\abs{g_{j}^{(0)}}$. This is a 1-sparse,
unbiased estimator of $g^{(0)}$ with $O(1)$ variance. Define $z^{(0)}:=\tilde{g}^{(0)}$.
Henceforth, for every $t\geq0,$ the subgradient descent step suggests
a direction $e^{(t)}$ in which to move whose sparsity is at most
the sparsity of $\tilde{g}^{(t)}.$ We partition $e^{(t)}=e_{+}^{(t)}+e_{-}^{(t)}$
into its positive and negative components. We then apply Lemma \ref{lem:subgradient-update_additive-1}
twice: once with $x=x^{(t)}$,$e=e_{+}^{(t)}$, and $\ell=t$, to
obtain random vector $z_{+}^{(t)}$of sparsity $t$, and then with
$x=x^{(t)}+e_{+}^{(t)}$, $e=e_{-}^{(t)}$, and $\ell=t,$ to obtain
the random vector $z_{-}^{(t)}$ of sparsity $t$. The estimate of
the gradient at time $t$ is the sum of these random vectors. That
is, for all $t\geq1,$ define $\tilde{g}^{(t)}:=\sum_{s\leq t}(z_{+}^{(s)}+z_{-}^{(s)})$.
By the property (b) of Lemma \ref{lem:subgradient-update_additive-1}
, $\tilde{g}^{(t)}$ is a valid stochastic subgradient and can be
fed into the framework of Lemma \ref{lem:framework}. Note that for
any $t\in[T]$, the sparsity of $\tilde{g}^{(t)}$is $O(t^{2})$ and
so is the sparsity of $e^{(t)}$ suggested by the stochastic subgradient
routine. Thus, the $t$th step of estimating $z_{+}^{(t)}$and $z_{-}^{(t)}$
requires time $O(t^{2}\time\log^{2}n)$, implying we can run $T$
steps of the above procedure in $O(T^{3}\time\log^{2}n)$ time. 

Finally, to argue about the number of iterations required to get $\eps$-close,
we need to upper bound $\E[\norm{\tilde{g}^{(t)}}_{2}^{2}]$ for every
$t$. Since $\E[\tilde{g}^{(t)}]=g^{(t)}$, the true subgradient at
$x^{(t)}$ and since $\norm{g^{(t)}}_{2}^{2}=O(1)$ by Lemma \ref{lem:bounded-subgradients}
, it suffices to upper bound $\E[\norm{\tilde{g}^{(t)}-\E[\tilde{g}^{(t)}]}_{2}^{2}]$.
But this follows since $\tilde{g}^{(t)}$ is just a sum of independent
$z$-vectors.
\[
\E[\norm{\tilde{g}^{(t)}-\E[\tilde{g}^{(t)}]}_{2}^{2}]=\sum_{s\leq t}\E[\norm{z_{+}^{(s)}-\E[z_{+}^{(s)}]}_{2}^{2}]+\sum_{s\leq t}\E[\norm{z_{-}^{(s)}-\E[z_{-}^{(s)}]}_{2}^{2}]=O\left(\sum_{s\leq t}1/s\right)=O(\log n)
\]

The second-last inequality follows from (c) of Lemma \ref{lem:subgradient-update_additive-1}.
And so, $\E[\norm{\tilde{g}^{(t)}}_{2}^{2}]=\E[\norm{\tilde{g}^{(t)}-\E[\tilde{g}^{(t)}]}_{2}^{2}]+\norm{g^{(t)}}_{2}^{2}=O(\log n)$.
Therefore, we can apply the framework in Lemma \ref{lem:framework}
with $B=O(\log n)$ implying the total number of steps to get $\eps$-approximate
is $N=O(n\log^{2}n\eps^{-2})$. Furthermore, since each batch takes
time $O((n+T^{3})\time\log^{2}n)$ and there are $N/T$ batches, we
get that the total running time is at most 
\[
O\left(n\time\log^{4}n\eps^{-2}\left(\frac{n+T^{3}}{T}\right)\right)=\tilde{O}(n^{5/3}\eps^{-2}\time)
\]
 if $T=n^{1/3}.$ This ends the proof of Theorem \ref{thm:runtime_additive}.

\end{proof}

\subsection{Improvements when Minimizer is Sparse \label{sec:fast_sfm:sparse}}

Here we discuss how to improve our running times when the submodular
function $f$ is known to have a sparse solution, that is, the set
minimizing $f(S)$ has at most $s$ elements. Throughout this section
we suppose we know $s$. 

\begin{thm}
\label{thm:sparse} Let $f$ be a submodular function with a $s$-sparse
minimizer. Then if $f$ is integer valued with $|f(S)|\leq M$ for
all $S\subseteq[n]$ we can compute the minimizer deterministically
in time $O((n+sM^{3})\log n\cdot\time)$. Furthermore if $f$ is real
valued with $|f(S)|\leq1$ for all $S\subseteq[n]$, then there is
a randomized algorithm which in time $\tilde{O}((n+sn^{2/3})\time\eps^{-2})$
returns a set $S$ such that $\E[f(S)]\le\OPT+\eps,$ for any $\eps>0$.
\end{thm}
Therefore, if we know that the sparsity of the optimum solution is,
say ${\tt polylog}(n)$, then there is a near linear time approximate
algorithm to get constant additive error. 

To obtain this running time we leverage the same data structures for
maintaining subgradients presented in Section~\ref{sec:fast_sfm:pseudopoly}
and Section~\ref{sec:fast_sfm:subquad}. Instead we show how to specialize
the framework we used presented in Section~\ref{sec:fast_sfm:framework}.
In particular we simply leverage that rather than minimizing the Lovasz
extension over $[0,1]^{n}$ we can minimize over $S_{s}\defeq\{x\in[0,1]^{n}\,|\,\sum_{i\in[n]}x_{i}\leq s\}$.
This preserves the value of the maximum and minimum, but now improves
the convergence of projected (stochastic) subgradient descent (because
the quantity $R$ becomes $s$ from $n$). To show this formally we
simply need to show that the projection step doesn't hurt the performance
of our algorithm asymptotically.

We break the proof of this into 3 parts. First, in Lemma~\ref{lem:projection_step}
we compute how to projection onto $S_{s}$. Then in Lemma~\ref{lem:sparse_framework}
we show how to update our framework. Using these, we prove Theorem~\ref{thm:sparse}.
\begin{lem}
\label{lem:projection_step} For $k\geq0$ and $y\in\R^{n}$ let $S=\{x\in[0,1]^{n}\,|\,\sum_{i}x_{i}\leq k\}$
and 
\[
z=\argmin_{x\in S}\frac{1}{2}\norm{x-y}_{2}^{2}.
\]
Then, we have that for all $i\in[n]$ 
\[
z_{i}=\text{median}(0,y_{i}-\lambda,1)
\]
where $\lambda$ is the smallest non-negative number such that $\sum_{i}z_{i}\leq k$. \end{lem}
\begin{proof}
By the method of Lagrange multiplier, we know that there is $\lambda\geq0$
such that
\[
z=\argmin_{x\in[0,1]^{n}}\frac{1}{2}\norm{x-y}_{2}^{2}+\lambda\sum_{i\in[n]}x_{i}\,.
\]
Since each variable in this problem is decoupled with each other,
we can solve this problem coordinate-wise and get that for all $i\in[n]$
\[
z_{i}=\text{med}(0,y_{i}-\lambda,1).
\]
Since $\sum_{i\in[n]}z_{i}$ decreases as $\lambda$ increases, we
know that $\lambda$ is the smallest non-negative number such that
$\sum_{i\in[n]}z_{i}\leq k$.
\end{proof}
In particular we provide Lemma~\ref{lem:sparse_framework} and improvement
on Lemma \ref{lem:framework}. 
\begin{lem}
\label{lem:sparse_framework} Suppose that for $N\geq sB^{2}\epsilon^{-2}$
and any sequence of $x^{(1)},..,x^{(N)}$ such that that $x^{(i+1)}-x^{(i)}$
is $k^{(i)}$-sparse up to modifications that do not affect the additive
distance between non-zero coordinates with $K\defeq\sum_{i\in[T]}k^{(i)}=O(Nk)$
we can implement a subgradient oracle for $f$ and $x^{(i)},$ denoted
$\tilde{g}(x^{(i)})$, that is $k$-sparse and obeys $\E\norm{\tilde{g}}_{2}^{2}\leq B^{2}$.
Then in time $O(n(\runtime+\log n)+Nk\log n)$ we can compute a set
$S$ such that $\E f(S)\leq\OPT+\epsilon$ (and if the subgradient
oracle is deterministic then the result holds without the expectation).\end{lem}
\begin{proof}
The proof is the same as before, just the size of $R$ improves to
$s$ and we need to deal with this new projection step. However, in
the projection step we set all the coordinates that are less than
$0$ to $0$ and then keep subtracting uniformly (stopping whenever
a coordinate reaches 0) until the maximum coordinate is $\leq1$.
We can do this efficiently by simply maintaining an additive offset
and the coordinate values in sorted order. Then we simply need to
know the number of coordinates above some threshold and the maximum
and the minimum non-zero coordinate to determine what to subtract
up to the point we make the minimum non-zero. We can do this in $O(\log n)$
easily. Now we are not counting the movements that do not set something
to 0 so do not change the additive distances between the non-zero
coordinate. Consequently, an iteration may only move many coordinates
if it sets many things to 0, however that is paid for by the movement
that created it, so we only need $Nk\log n$ time in total to do all
the updates.
\end{proof}
We now have everything we need to prove Theorem~\ref{thm:sparse}
\begin{proof}[Proof of Theorem~\ref{thm:sparse}]
 (Sketch) The proof is the same as in Section~\ref{sec:fast_sfm:pseudopoly}
and Section~\ref{sec:fast_sfm:subquad}. We just use Lemma~\ref{lem:sparse_framework}
instead of the previous framework lemma. To invoke the first data
structure just do the update in batches. the second data structure
was already written for this setting. \end{proof}

%% file: lower_bound.tex
\section{Lower Bound}

It is well known that $\Omega(n)$ evaluation oracle calls are needed
to minimize a submodular function. On the other hand, the best way
we know of for certifying minimality takes $\Theta(n)$ subgradient
oracle calls (or equivalently, vertices of the base polyhedron). A
natural question is whether $\Theta(n)$ subgradient oracle calls
are in fact needed to minimize a submodular function. In this section
we answer this in the affirmative. Since each gradient oracle needs
$n$ evaluation oracle calls, this gives an $\Omega(n^{2})$ lower
bound on the number of evaluations required for algorithms which only
access the function via graident oracles. As mentioned in the introduction,
these include the Fujishige-Wolfe heuristic \cite{Wolfe76,Fujishige80},
various version of conditional gradient or Franke Wolfe \cite{frankWolfe,linearConvergentCondGrad}
, and the new cutting plane methods \cite{LSW15}. Note that there
are known lower bounds for subgradient descent that have a somewhat
submodular structure \cite{Nesterov2003} and this suggests that such
a lower bound should be possible, however we are unaware of a previous
information theoretic lower bound such as we provide.

To prove our lower bound, we describe a distribution over a collection
of hard functions and show that any algorithm must make $\Omega(n)$
subgradient calls in expectation\footnote{One can also prove a high probability version of the same result but
for simplicity we don't do it. } and by Yao's minimax principle this will give an $\Omega(n)$ lower
bound on the expected query complexity of any randomized SFM algorithm.
The distribution is the following. Choose $R$ to be a random set
with each element of the universe selected independently with probability
$1/2$. Given $R$, define the function

\[
f_{R}(S)=\begin{cases}
-1 & \text{if }S=R\\
0 & \text{if }S\subsetneq R\text{ or }R\subsetneq S\\
1 & \textrm{otherwise}.
\end{cases}
\]

Clearly the minimizer of $f_{R}$ is the set $R.$ Any SFM algorithm
is equivalent to an algorithm for recognizing the set $R$ via subgradient
queries to $f_{R}$. A subgradient $g$ of $f_{R}$ at any point $x$
corresponds to a permutation $P$ of $\{1,2,\ldots,n\}$ (the sorted
order of $x$). Recall the notation $P[i]:=\left\{ P_{1},P_{2},\ldots,P_{i}\right\} $.
The following claim describes the structure of subgradients. 
\begin{lem}
\label{lem:sub_r}Let $i$ be the smallest index such that $P[i]$
\textbf{is} \textbf{not} a subset of $R$ and $j$ be the smallest
index such that $P[j]$ \textbf{is} a superset of $R$. Then $g(i)=1,$
$g(j)=-1$, and $g(k)=0$ for all $k\in[n]\setminus{i,j}.$\end{lem}
\begin{proof}
To see $g(i)=1$, note that $A:=P[i-1]$ is a subset of \textbf{$R$.
}Two cases arise: either $A=R$ in which case $P[i]$ is a strict
superset of $R$\textbf{ }and so $f_{R}(A)=-1$ and $f_{R}(P[i])=0$
implying $g(i)=1$; or $A$ is a strict subset of $R$ in which case
$P[i]$ is neither a subset or a superset, implying $f_{R}(A)=0$
and $f_{R}(P[i])=1$. Similarly, to see $g[j]=-1,$ note that $B:=P[j-1]$
is not a superset of $R$. Two cases arise: either $B$ is a strict
subset of $R$ in which case $P[j]=R$ and we have $f_{R}(B)=0$ and
$f_{R}(P[j])=-1$; or $B$ is neither a subset nor a superset in which
case $P[j]$ is a strict superset of $R$ and we have $f_{R}(B)=1$
and $f_{R}(P[j])=0$. 

For any other $k,$ we have either both $P[k]$ and $P[k-1]$ are
strict subsets of $R$ (if $k<\min(i,j)$), or both $P[k]$ and $P[k-1]$
are strict supersets of $R$ (if $k>\max(i,j)$) , or both are neither
superset nor subset. In all three cases, $g(k)=0\mbox{. }$
\end{proof}
Intuitively, any gradient call gives the following information regarding
$R$: we know elements in $P[i-1]$ lie in $R,$ $P_{i}$ doesn't
lie in $R,$ $P_{j}$ lies in $R,$ and all $P_{k}$ for $k>j$ do
not lie in $R.$ Thus we get $i+n-j+1$ ``bits'' of information.
If $R$ is random, then the expected value of this can be shown to
be $O(1)$, and so $\Omega(n)$ queries are required. We make the
above intuitive argument formal below.

Suppose at some point of time, the algorithm knows a set $A\subseteq R$
and a set $B\cap R=\emptyset.$ The following lemma shows that one
may assume wlog that subsequent subgradient calls are at points $x$
whose corresponding permutation $P$ contains the elements of $A$
as a ``prefix'' and elements of $B$ as a ``suffix''. 
\begin{lem}
\label{lem:known}Suppose we know $A\subseteq R$ and $B\cap R=\emptyset$.
Let $g$ be a subgradient and $g'$ be obtained from $g$ by moving
$A$ and $B$ to the beginning and end of the permutation respectively.
Then one can compute $g$ from $g'$ without making any more oracle
calls.\end{lem}
\begin{proof}
Easy by case analysis and Lemma \ref{lem:sub_r}. Let $P$ be the
permutation corresponding to $g$. We show that given $g$' and $P$,
we can evaluate $g.$ Let us say we are interested in evaluating $g_{P_{k}}$
and say $P_{k}=a.$ Lemma \ref{lem:sub_r} states that this is 1 iff
$P[k-1]\subseteq R$ and $P[k]$ isn't. Now, if $P[k-1]\cap B\neq\emptyset,$
then we know $g_{P_{k}}=$0. Otherwise, $g_{P_{k}}=$1 iff $P[k-1]\setminus B\cup A\subseteq R$
and $P[k]\setminus B\cup A$ is not, since $A\subseteq R.$ Therefore,
$g_{P_{k}}=$1 iff $g'_{a}=1$ and $P[k-1]\cap B=\emptyset.$ Whether
$g_{P_{k}}=-1$or not can be done analogously. 
\end{proof}
For an algorithm, let $h(k)$ be the expected number of subgradient
calls required to minimize $f_{R}$ when the universe if of size $k$
(note $R$ is chosen randomly by picking each element with probability
1/2). For convenience we also define $h(k)=0$ for $k\leq0$.
\begin{lem}
\label{lem:recur}For $k\geq1$, $h(k)\geq1+\mathbb{\E}_{X,Y}[h(k-X-Y)]$,
where $X,Y$ are independent geometric random variables, i.e. $Pr[X=i]=1/2^{i}$
for $i\geq1$.\end{lem}
\begin{proof}
By our observation above, a subgradient of $f$ reveals the identities
of $\min\{X+Y,k\}$ elements, where $X-1=i-1$ and $Y-1=n-j$ ($i,j$
as defined in Lemma \ref{lem:sub_r}) are the lengths of the streaks
of 0's at the beginning and end of the subgradient.

Note that $X$ simply follow a geometric distribution because $Pr[P[i-1]\subseteq R,P_{i}\notin R]=1/2^{i}$.
Similarly, $Y$ also follow the same geometric distribution. In the
case of $X+Y>k$, we have $R$ as a prefix of the permutation.

Finally, as a subgradient call reveals no information about the intermediate
elements in the permutation, by Lemma \ref{lem:known} we are then
effectively left with the same problem of size $k-X-Y$. More formally,
this is because the value of the subgradient queried is independent
of the identities of the elements $P_{i+1},\ldots,P_{j-1}$.\end{proof}
\begin{thm}
$h(n)\geq n/4$, i.e. any algorithm for SFM requires at least $\Omega(n)$
subgradient calls.\end{thm}
\begin{proof}
We show by induction that $h(k)\geq k/4$. By Lemma \ref{lem:recur}
and the induction hypothesis,
\begin{eqnarray*}
h(k) & \geq & 1+\mathbb{\E}_{X,Y}[h(k-X-Y)]\\
 & \geq & 1+\mathbb{\E}_{X,Y}[(k-X-Y)/4]\\
 & = & 1+k/4-\mathbb{\E}[X]/4-\mathbb{\E}[Y]/4\\
 & = & k/4
\end{eqnarray*}
as desired.
\end{proof}
Readers may have noticed that the proofs of the preceding two lemmas
essentially imply that $h(k)$ is roughly the expected number of geometric
random variables needed to sum up to $k$. One can use this property
together with some concentration inequality for geometric random variables
to establish a high probability version of our lower bound.

%% file: appendix.tex
\section{Reduction from Multiplicative to Additive Approximation}

\label{app:mult_to_add}

Here we show how to obtain a multiplicative approximation for SFM
from our $\tilde{O}(n^{5/3}\cdot\time/\varepsilon^{2})$ additive-approximate
SFM algorithm. Because the minimizer of $f$ is scale- and additive-invariant,
it is necessary to make certain regularity assumptions on $f$ to
get a nontrivial result. This is akin to submodular function maximization
where constant factor approximation is possible only if $f$ is nonnegative
everywhere \cite{BFNS12,FMV07}. For SFM, by considering $f-\OPT$
we see that finding a multiplicative-approximate solution and an exact
solution are equivalent for general $f$. (Indeed most submodular
optimization problems permit multiplicative approximation only in
terms of the range of values.)

Similar to submodular maximization, we assume $f$ to be \textit{nonpositive}.
Then $f'=f/\OPT$ has range $[-1,0]$ and has minimum value -1 so
our additive-approximate algorithm immediately yields multiplicative
approximation. This requires knowing $\OPT$ (or some constant factor
approximation of). Alternately we can ``binary search'' to get factor-2
close to $\OPT$ by trying different powers of 2. This would lead
to a blowup of $O(\log\OPT)$ in the running time.

\section{Approximate SFM via Fujishige-Wolfe }

\label{app:approx_sfm_wolfe}

Here we show how Frank-Wolfe and Wolfe can give $\eps$-additive approximations
for SFM. We know that both algorithms in $O(1/\delta)$ iterations
can return a point $x\in B_{f}$, the base polyhedron associated with
$f$, such that $x^{\top}x\leq p^{\top}p+\delta$ for all $p\in B_{f}$.
Here we are using the fact implied by Lemma \ref{lem:bounded-subgradients}
that the diameter of the base-polytope for functions with bounded
range is bounded (note that vertices of the base polytope correspond
to gradients of the Lovasz extension.) The robust Fujishige Theorem
(Theorem 5, \cite{CJK14}) implies that we can get a set $S$ such
that $f(S)\le\OPT+2\sqrt{n\delta}.$ Setting $\delta=\eps^{2}/4n$
gives the additive approximation in $O(n\eps^{-2})$ gradient calls.

\section{Faster Algorithm for Directed Minimum Cut}

\label{sec:minimum_cut}

Here we show how to easily obtain faster approximate submodular minimization
algorithms in the case where our function when the funciton is an
explicitly given $s$-$t$ cut function. This provides a short illustration
of the reasonable fact that when given more structure, our sumbodular
minimization algorithms can be improved.

For the rest of this section, let $G=(V,E,w)$ be a graph with vertices
$V$, directed edges $E\subseteq V\times V$, and edge weights $w\in\R_{\geq}^{E}0$.
Let $s,t\in V$ be two special vertices, $A\defeq V\setminus\{s,t\}$,
and for all $S\subseteq A$ let $f(S)$ be defined as the total weight
of the edges in leaving the set $S\cup\{s\}$, i.e. where the tail
of edge is in $S\cup\{s\}$ and the head of the edges is in $V\setminus(S\cup\{s\})$.
The function $f$ is a well known submodular function and minimizing
it corresponds to computing the minimum $s$-$t$ cut, or correspondingly
the maximum $s$-$t$ flow. 

Note that clearly, $f(S)\leq W$ where $W=\sum_{e\in E}w_{e}$. Furthermore,
if we pick $S$ by including each vertex in $A$ randomly to be in
$S$ with probability independently $\frac{1}{2}$ then we see that
$\E f(S)=\frac{1}{2}W$. Consequently, $\frac{1}{2}W\leq\max_{S\subseteq A}f(s)\leq W$
and if we want to scale $f$ to make it have values in $[-1,1]$ we
need to devide by something that is $W$ up to a factor of $2$.

Now, note that we can easily extend this problem to a continuous problem
over the reals. Let $x^{+}$ denote $x$ if $x\ge0$ and $0$ otherwise.
Furthermore, for all $x\in\R^{A}$ let $y(x)\in\R^{V}$ be given by
$y(x)_{i}=x_{i}$ if $i\in A$, $y(x)_{s}=0$, $y(x)_{t}=1$, and
let 
\[
g(x)\defeq\sum_{(a,b)\in E}w_{ab}(y(x_{b})-y(x_{a}))^{+}\,.
\]
Clearly, minimizing $g(x)$ over $[0,1]^{A}$ is equivalent to minimizing
$f(S)$. Furthermore the subgradient for $g$ decomposes into subgradients
for each edge $(a,b)\in E$ each of which is a vector with 2 non-zero
entries and norm at most $O(w_{ab})$. If we picking a random edge
with probability proportional to $w_{ab}$ and output its subgradient
scaled by $W/w_{ab}$ subgradient this yields a stochastic subgradient
oracle $\tilde{g}(x)$ with $\E\norm{\tilde{g}(x)}_{2}^{2}=O(\sum_{(a,b)\in E}\frac{w_{ab}}{W}((W/w_{ab})\cdot w_{ab})^{2})=O(W^{2})$.
Consequently, by Theorem~\ref{thm:subgradient_descent} setting $R^{2}=O(|V|)$
we see that we can compute $z$ with $g(z)-\min_{x}g(x)\leq W\epsilon$
in $O(|v|\epsilon^{-2})$. Thus, if we scaled $g$ to make it $[-1,1]$
valued the time to compute an $\epsilon$-approximate solution would
be $O(|V|\epsilon^{-2})$.

This shows that an explicit instance of minimum $s$-$t$ cut does
not highlight the efficacy of the approach in this paper. Instantiating
our algorithm naively would give an $\otilde(|E|\cdot|V|^{5/3}\cdot\epsilon^{-2})$
to achieve additive error $\epsilon$. Nevertheless, even for such
an instance if instead we were simply given access to the an $\time$
time evaluation oracle for $f$, and the graph was desne, even in
this instance, without knowing the structure aprior we do not know
how to improve upon the $O(\time\cdot|V|^{5/3}\epsilon^{-2})$ time
bound achieved in this paper (though no serious attempt was made to
do this). In short there may be a gap between explicitly given structured
instances of submodular functions and algorithms that work with general
evaluation oracles as focused on in this paper.

\section{Certificates for Approximate SFM }

\label{app:approx_sfm_certificates}

The only certificate we know to prove that the optimum value of SFM
is $\geq F$ is to show a certain vector $x$ lies in the base polyhedron.
For example, one proof via Edmond's Theorem \cite{E70} is by demonstrating
$x\in B_{f}$ whose negative entries sum to $\geq F$. The only way
to do this is via Carathedeory's Theorem which requires $n$ vertices
of $B_{f}$, each of which requires $n$ function evaluations. For
approximate SFM, one thought might to be to use approximate Caratheodory's
Theorems \cite{Berman15,MPVW15} to describe a nearby point $x'$.
Unfortunately, for $\eps$-additive SFM approximation, one needs $x'$and
$x$ to be close in $\ell_{1}$-norm and approximate Caratheodory
works only for $\ell_{2}$-norm and higher. If one uses the $\ell_{2}$-norm
approximation, then unfortunately one doesn't get anything better
than quadratic. More precisely, approximate Caratheodory states that
one can obtain $||x'-x||_{2}\leq\delta$ with support of $x'$ being
only $O(1/\delta^{2})$-sparse. But to get $\ell_{1}$ approximations,
we need to set $\delta=\eps\sqrt{n}$ leading to linear sized support
for $x'.$ The approximate Caratheodory Theorems are tight \cite{MPVW15}
for general polytopes. Whether one can get better theorems for the
base polyhedron is an open question.

\section{Pseudocodes for Our Algorithms}

\label{app:algorithm_pseudocode}

We provide guiding pseudocodes for our two algorithms.

\label{alg:exact}
\begin{algorithm}[H]
\label{alg:exact-1}\textbf{Initialization.}
\begin{itemize}
\item \setlength{\itemsep}{-1mm}$x^{(1)}\defeq0^{n}$
\item Evaluate $g^{(1)}$ is the Lovasz subgradient at $x^{(1)}$. \emph{(Takes
$O(n\cdot\time)$ time. Store as (coordinate, value) pair in set $S^{(1)}.$
$|S^{(1)}|\leq3M.$ )}
\item Store $x^{(1)}$ in a balanced Binary search tree. At each node store
the \textbf{value} that is the sum of the gradient coordinates corr.
to children in the tree. \emph{(Takes $O(n)$ time to build.)}
\item Set $T\defeq20nM^{2}.$ Set $\eta\defeq\frac{\sqrt{n}}{18M}$.
\end{itemize}
\textbf{For }$t=1,2,\ldots,T:$ 
\begin{itemize}
\item \setlength{\itemsep}{-1mm}Define $e^{(t)}$ which is non-zero in
coordinates corresponding to $S^{(t)}$\emph{: (Takes time $|S^{(t)}|\le3M$
.)}

\begin{itemize}
\item if $g_{i}^{(t)}>0,$ then $e_{i}^{(t)}=\min(x_{i}^{(t)},\eta g_{i}^{(t)})$
\item if $g_{i}^{(t)}<0,$ then $e_{i}^{(t)}=\max(x_{i}^{(t)}-1,\eta g_{i}^{(t)})$
\end{itemize}
\item \textbf{Update}($x^{(t)},e^{(t)},S^{(t)}$) to get $(x^{(t+1)},g^{(t+1)},S^{(t+1)}$)
where $g^{(t+1)}$ is stored as coordinate,value pairs in $S^{(t+1)}.$
as described in Lemma \ref{lem:subgradient-update}. \emph{(Update
takes time $O(M\log n+M\cdot\time+M\cdot\time\log n)$)}
\end{itemize}
Obtain the $O(n)$ sets given\textbf{ }the order\textbf{ }of \textbf{$x_{T}$},
that is, if $P$ is the permutation corresponding to $x_{T}$, then
the sets are $\{P[1],\ldots,P[n]\}$. Return the minimum valued set
among them.\textbf{ }

\caption{Near Linear Time Exact SFM Algorithm. }
\end{algorithm}

\begin{algorithm}[H]
\textbf{Initialization}
\begin{itemize}
\item \setlength{\itemsep}{-1mm}Set $N\defeq10n\log^{2}n\eps^{-2},$ $T=\ceil{n^{1/3}}$
\item Initialize $x$ as the all zeros vector and store it in a BST. 
\end{itemize}
For $i=1,2,\ldots N/T:$
\begin{itemize}
\item \setlength{\itemsep}{-1mm}$x^{(1)}\defeq$the current $x$.
\item Compute $g^{(1)}$, the gradient to the Lovasz extension given $x^{(1)}.$
\emph{//This takes $O(n\time)$ time).}
\item Sample $z^{(1)}$ by picking $j\in[n]$ with probability proportional
to $\abs{g_{j}^{(1)}}$ and returning $z^{(1)}\defeq\norm{g^{(1)}}_{1}sign(g_{j}^{(1)})\cdot\mathbf{1}_{j}$.
\emph{//This takes $O(n\time)$ time.}
\item Set $\tilde{g}^{(1)}\defeq z^{(1)}.$
\item For $t=1,2,\ldots,T:$

\begin{itemize}
\item \setlength{\itemsep}{-1mm}Define $e^{(t)}$ as in Algorithm 1 using
$\tilde{g}^{(t)}$ instead of $g^{(t)}.$\emph{ //This takes time
$O(\supp(\tilde{g}^{(t)}))$ which will be $O(t^{2})$}
\item Obtain $z^{(t)}$ using \textbf{Sample}($x^{(t)},e^{(t)},\ell=t)$
where \textbf{Sample }is the randomized procedure describe in Lemma
\ref{lem:subgradient-update_additive-1}. \emph{//This takes $O(t^{2}\time\log n)$
time.}
\item Update $\tilde{g}^{(t+1)}\defeq\sum_{s\leq t}z^{(s)}$. \emph{//This
takes $O(t^{2}\log n)$ time to update the relevant BSTs.}
\end{itemize}
\item Set current $x$ to $x_{T}.$
\end{itemize}
Obtain the $O(n)$ sets given\textbf{ }the order\textbf{ }of the final
$x$, that is, if $P$ is the permutation corresponding to $x$, then
the sets are $\{P[1],\ldots,P[n]\}$. Return the minimum valued set
among them.\textbf{ }

\caption{Subquadratic Approximate SFM Algorithm.}
\end{algorithm}